\documentclass[11pt]{article}

\usepackage{listings}
\usepackage{url}

\usepackage[english]{babel}
\usepackage{setspace}
\usepackage[letterpaper,top=2cm,bottom=2cm,left=3cm,right=3cm,marginparwidth=1.75cm]{geometry}
\usepackage[dvipsnames,table,xcdraw]{xcolor}

\newcommand{\bigconcchoose}[1]{\def\bigconcsize{}%
  \ifx#1\displaystyle
    \let\bigconcsize\Big
  \else
    \ifx#1\textstyle
      \let\bigconcsize\big
    \fi
  \fi#1}

\newcommand{\bigxor}{\mathop{\mathchoice
  {\textstyle\bigoplus}{\textstyle\bigoplus}
  {\scriptstyle\bigoplus}{\scriptscriptstyle\bigoplus}}}

\usepackage{amsmath}
\usepackage{graphicx}
\usepackage{hyperref} 
\hypersetup{
    colorlinks,
    citecolor=blue,
    filecolor=black,
    linkcolor=blue,
    urlcolor=blue
}
\usepackage{amsfonts}
\usepackage{amsthm}
\usepackage{amssymb}

\newtheorem{lemma}{Lemma}
\newtheorem{theorem}[lemma]{Theorem}

\newtheorem{corollary}[lemma]{Corollary}
\newtheorem{definition}[lemma]{Definition}

\usepackage[style=numeric, sorting=none]{biblatex}
\title{Fast Simulation of Cellular Automata by Self-Composition}

\author{Joseph Natal$^{\dag}$, Oleksiy Al-saadi$^{\ddagger}$  \\
        \small $^{\dag}$Karlsruhe Institute of Technology, {\nolinkurl{joseph.natal@kit.edu}} \\
        \small $^{\ddagger}$Sonoma State University, {\nolinkurl{alsaadio@sonoma.edu}} \\
}

\begin{document}

\maketitle
\begin{abstract}
Computing the configuration of any one-dimensional cellular automaton at generation $n$ can be accelerated by constructing and running a composite rule with a radius proportional to $\log n$. The new automaton is the original one, but with its local rule function composed with itself. Consequently, the asymptotic time complexity to compute the configuration of generation $n$ is reduced from $O(n^2)$-time to $O(n^2 / \log n)$, but with $O(n^2 / (\log n)^3)$-space, demonstrating a time-memory tradeoff. Experimental results are given in the case of Rule $30$.
\end{abstract}


% The text of the paper follows. All of the text should be in the same file. 
% Use separate files for large tabular material and graphics.

\section{Introduction}
\label{intro}
Compositions of cellular automata have been discussed in different contexts since their inception decades ago \cite{Wolfram1984}, but their relevance for time complexity improvements has yet to be determined using formal asymptotic analysis. Indeed, compositions were explored by Israeli and Goldenfeld \cite{PhysRevE.73.026203} as so-called ``coarse-graining'', and were shown to have implications for understanding emergent phenomena in complex systems. Riedel and Zenil \cite{RiedelUniversality,RiedelComposite} further explored coarse-graining, finding subsets of elementary cellular automata which emulate all others, giving evidence of a pervasive universality. Perhaps the most similar algorithm in concept to the one described herein is given by Gosper's hashed implementation of the Game of Life \cite{GOSPER198475}. This algorithm memoizes local recurrent space-time patterns for Life-like automata, but performs poorly for chaotic structures. Our focus on Rule $30$ in this paper stems from its notoriously chaotic nature and perceived lack of structure. The naive method of computing the configuration at generation $n$ of an ECA is simply to update row by row, requiring $O(n^2)$-time. In contrast, CAs with nested patterns such as the well-studied Rule $90$ or Rule $150$ contain a fractal structure which allows for a given space-time coordinate to be determined in $O(\log n)$ \cite{Culikfractal}.

\section{Preliminaries}

A \emph{cellular automaton} (CA) is composed of a regular grid of colored cells whose colors are updated according to certain rule which may be described by a \emph{rule icon}. For example, the complete set of rule icon of the well-studied Rule $30$ cellular automaton is as follows:

\begin{center}
\begin{tabular}{ c c c c c c c c}
 $\blacksquare\!\blacksquare\!\blacksquare$  & $\blacksquare\!\blacksquare\!\square$ & $\blacksquare\!\square\!\blacksquare$ & $\blacksquare\!\square\!\square$ & $\square\!\blacksquare\!\blacksquare$ & $\square\!\blacksquare\!\square$ & $\square\!\square\!\blacksquare$ & $\square\!\square\!\square$ \\ 
 $\square$ & $\square$ & $\square$ & $\blacksquare$ & $\blacksquare$ & $\blacksquare$ & $\blacksquare$ & $\square$
\end{tabular}
\end{center}

\noindent The rule icon is applied to a row of cells in a one-dimensional CA by iterating through every cell in the row, comparing the colors of a cell along with its two adjacent neighbors, and then updating the color of the cell in the succeeding row (see Fig. \ref{fig:composite_rules_spacetime}). Let $\mathcal F$ be a CA with grid colors indexed by the integers $\mathbb Z_k = \{0,1,\hdots, k-1\}$, a given initial state, and local update rule $f : \mathbb Z_k^{2r+1} \rightarrow \mathbb Z_k$. Here, $r$ is a positive integer that denotes the \emph{radius} of the CA: the number of neighbor cells to the left and right of the current cell which are taken into account when applying the rule icon. These cells are sometimes referred to as a \emph{neighborhood}. For example, the rule icon of Rule $30$ shows that the rule has radius $1$ (see the above rule icon). 

We will restrict our study of CA to those with two colors by defining $\Sigma = \mathbb Z_2 = \{0,1\} = \{\square,\blacksquare\}$. A CA will subsequently and informally refer to one that is one-dimensional and $2$-color. If a $2$-color CA $\mathcal F$ has radius $1$ then we call $\mathcal F$ an \emph{elementary cellular automaton} (ECA). Wolfram popularized a systematic numbering scheme for ECA where for any given Rule $k$ the corresponding rule icon can be obtained by taking the binary digits of integer $k$ and assigning them as outputs to the set of binary $3$-tuples: the numbers $0$ through $7$ in decimal. Using the above example rule icon, the output cells are, sequentially, $00011110_2$ which is $30_{10}$. There are $2^{2^3} = 256$ ECA, though many are mirrored and therefore exhibit the same behavior.

A \emph{configuration} $X = \langle x_{-(N-1)/2}, \dots ,x_{(N-1)/2} \rangle$ refers to a row of cells having $N \in 2 \mathbb N + 1$ indices. It will be convention that the median element in the array has index $i=0$. When necessary to distinguish, $X_n^{\mathcal F}$ denotes the bi-infinite configuration of $\mathcal F$ at generation $n$ (i.e. after $n-1$ applications of its local rule beyond the specified initial configuration). It follows that $X_n^{\mathcal F} = \langle \dots, x_{i-2}^n, x_{i-1}^n,x_i^n, x_{i+1}^n, x_{i+2}^n, \dots \rangle$. We explicitly define the \emph{global update} of a configuration as  
\begin{equation} \label{eq:global_update}
\begin{split}
      X_{n+1}^{\mathcal F} = \bigxor_{i=-\infty}^{\infty} f\big( \bigxor_{j=-r}^r x^n_{i+j}\big)
\end{split}
\end{equation}
and for an individual cell in $X_{n+1}^{\mathcal F}$
\begin{equation} \label{eq:global_update_single_index}
      x^{n+1}_i = f\big( \bigxor_{j=-r}^r x^n_{i+j} \big)
\end{equation}
where $\bigxor$ denotes concatenation of cells. Borrowing from physics terminology, $n$ reflects a dimension in time while spatial dimension is reflected by the $i$-th index of $X_n^{\mathcal F}$. A CA $\mathcal F$ has a \emph{simple seed} if $X_1^{\mathcal F} = \langle \cdots \square,\square,\square,\blacksquare,\square,\square,\square,\cdots \rangle$. 

The rule icons of any CA with radius $r$ can be equivalently represented through an algebraic form as a function $f : \Sigma^{2r+1}  \rightarrow \Sigma$. For example, for any generation $n+1$ and index $i$, the color $x^{n+1}_i$ of Rule $150$ can be expressly computed by the following trivariate function:

\begin{equation} \label{eq:150algebraic}
x_i^{n+1} = x_{i-1}^n + x_i^n + x_{i+1}^n  \mod 2
\end{equation}
The Rule $150$ CA is well known for its fractal self-symmetry and its algebraic form is similarly straightforward. Meanwhile, the following nonlinear discrete dynamical equation governs Rule $30$: 

\begin{equation} \label{eq:rule30nonlin}
x^{n+1}_i = \overbrace{x^{n}_{i-1} + x^{n}_{i} + x^{n}_{i+1} 
   }^\text{Rule 150} \;\;\; +\overbrace{x^{n}_{i} x^{n}_{i+1}}^\text{nonlinear term} \;\;  \mod 2
\end{equation}
or equivalently, using boolean algebra: 
\begin{equation} \label{eq:30bool}
x_i^{n+1} = x_{i-1}^n \oplus (x_i^n \lor x_{i+1}^n)
\end{equation}

\noindent Notice that the rule function of Rule $30$ contains the algebraic form of Rule $150$. Although this component of the function is rather simple in and of itself, the abrupt complexity of Rule $30$ arises from the nonlinear term. We introduce a special class of graphs that will be used for algorithmic analysis in the following lemmas.

\begin{definition}
    A De Bruijn graph is a directed graph representing overlaps between sequences of symbols. For a $2$-color CA $\mathcal H$ of radius $r$ and local rule $h$, its De Bruijn graph $B^{\mathcal H}$ will have $|\Sigma^{2r}|=2^{2r}$ vertices, each representing a cell neighborhood of length $2r$. The vertices each have $2$ outgoing edges corresponding to the $(2r+1)$-th cell color. These edges are directed to the vertices that represent the neighborhood realized after a unit shift of the original neighborhood. We associate a color with every edge that corresponds to the output of the rule function $h$ applied to the traversed neighborhood.
    \label{def:deB}
\end{definition}

The De Bruijn graphs have been well studied by Wolfram \cite{Wolfram1984} in the context of cellular automata. Fig. \ref{fig:state_transition_graph} details the De Bruijn graph of Rule 30, and Fig. \ref{fig:rule_30_compositions} shows that these graphs have a very regular structure. These graphs are used for DNA sequence assembly, and so optimization and compression of them is an active area of research \cite{10.1093/bioinformatics/btx067}. 
\section{Automata Self-Composition}

In this section, we show the local rule function of a cellular automata can be composed with itself in order to create a new CA satisfying special constraints on its configurations. 

\begin{definition} \label{def:composition}
Let the composition of two rules with functions $f_1$ and  $f_2$ be defined as 
\begin{align} \label{eq:composition}
  (f_2 \circ f_1)(X) =  f_2\big( \bigxor_{i=-r_2}^{r_2} f_1(\bigxor_{j=-r_1}^{r_1}x_{i+j})\big) 
\end{align}
This composite rule is therefore a local rule with radius $r_1 + r_2$, and a function mapping $ X \in \Sigma^{2 (r_1 + r_2) + 1} \rightarrow \Sigma$.
\end{definition}

By Definition \ref{def:composition}, the composite rule $f_2 \circ f_1$ consists of applying $f_1$ to all contiguous subarrays of length $2 r_1 + 1$ in $X$ with truncated boundary conditions in indexed order and then $f_2$ to the concatenated results. In short, the configuration is updated according to $f_1$, then $f_2$.

The next lemma is vital to the correctness of our algorithm.

\begin{lemma}
\label{lem:doubler} 
Given a CA $\mathcal H$ with local rule $h$ and radius $r$, there exists a CA $\mathcal G$ with local rule $g$ and radius $2r$ such that for every $n \in \mathbb{N}$ we have that $X_{2n-1}^{\mathcal H} = X_n^{\mathcal G}$. 
\end{lemma}
 
\begin{proof}
For self-composition, Eq. \ref{eq:composition} is reduced to 

\begin{align} \label{eq:compositionReduced}
    (h \circ h)(X) : \big\{\Sigma^{4r + 1} \rightarrow \Sigma : X \;\mapsto\; h\big( \bigxor_{i=-r}^{r} h(\bigxor_{j=-r}^{r}x_{i+j})\big) \big\} 
\end{align}
Replacing $h$ with $h\circ h$ in the global update (Eq. \ref{eq:global_update}) and assigning a generation $n$ gives
\begin{equation}
\label{eq:replace_h}
\begin{split}
    X^{\mathcal G}_{n+1} = \bigxor_{i=-\infty}^{\infty} (h \circ h)\big( \bigxor_{j=-2r}^{2r} x^n_{i+j}\big) = \bigxor_{i=-\infty}^{\infty} h\big( \bigxor_{j=-r}^{r} h(\bigxor_{k=-r}^{r}x^n_{i+j + k})\big)
\end{split}
\end{equation}
and substituting $x^{n+1}_{i+j} = h(\bigxor_{k=-r}^{r}x^n_{i+j + k})$ from Eq. \ref{eq:global_update_single_index}

\begin{equation} \label{eq:compositionEquivalent}
\begin{split}
     = \bigxor_{i=-\infty}^{\infty} h\big( \bigxor_{j=-r}^{r} x^{n+1}_{i+j}\big) = X^{\mathcal{H}}_{n+2}
\end{split}
\end{equation}
Then the composite update is simply two global updates according to $h$ in sequence. Hence, we construct $\mathcal G$ as follows: Let $g = h \circ h$ be the local rule function of $\mathcal G$ with the same initial configuration as $\mathcal H$ (i.e. $X_1^{\mathcal H} = X_1^{\mathcal G}$). Then we have that $X_{3}^{\mathcal H} = X_{2}^{\mathcal G}$ by Eq. \ref{eq:compositionEquivalent}. It follows that by repeated application of the global update for $\mathcal G$ for $n=1,2,3,\hdots$, the evolution for $\mathcal H$ will have $1, 3, 5, 7, \hdots = 2n-1$ and $X_{2n-1}^{\mathcal H} = X_n^{\mathcal G}$. Intuitively, running $\mathcal G$ for a single generation provides the same output configuration as running $\mathcal H$ for two generations (Fig. \ref{fig:composite_rules_spacetime}).
\end{proof}

\begin{table}[b]
    \centering
    \begin{tabular}{c|c|c}
        $k$-fold of Rule $30$ & $2r+1$ &  \textbf{Rule} \\
        \hline
         1&3 & $30$\\
         2&5 & $535945230$\\
         3&7 & $4245223 \hdots 81390$\\
         4&9 & $1672702 \hdots 88750$\\
    \end{tabular}
    \caption{Several rules are shown that can be constructed through self-composition of Rule $30$, along with their radii.}
    \label{tab:rule_comp_debruijn_debruijn}
\end{table}

\begin{figure}
    \centering
    \includegraphics[width=0.9\linewidth]{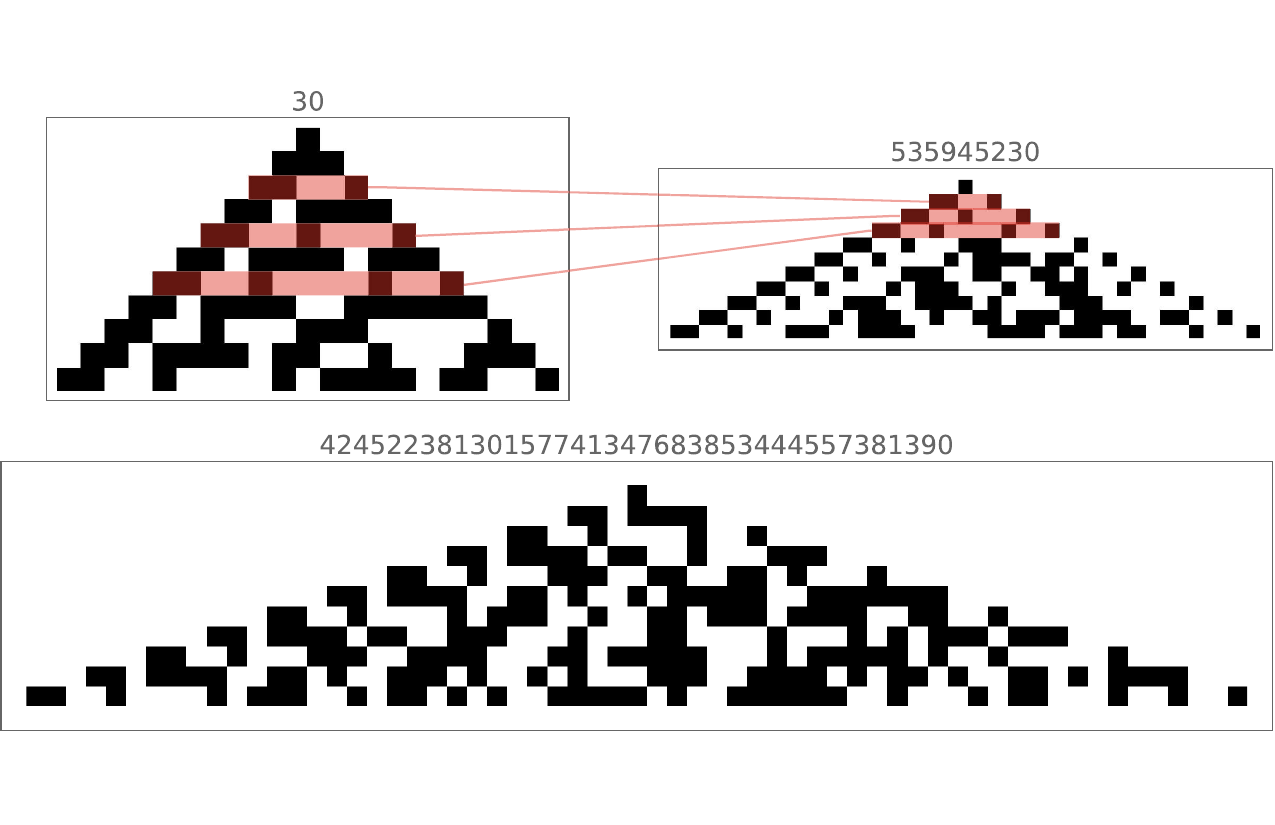}
    \caption{Evolution of Rule $30$ beginning with a simple seed and its $2$-fold and $3$-fold composition with itself (see Definition \ref{def:folding}). Highlighted rows show a sample of equivalent configurations.}
    \label{fig:composite_rules_spacetime}
\end{figure}

Now we extend this composition of automata to arbitrary $k$ and define a new terminology to describe CA that are of this type:
\begin{definition}
\label{def:folding}
Given a CA $\mathcal H$ with local rule $h$ and a CA $\mathcal F$ with local rule $f$, let $h_i = h$ for each $1 \leq i$. We say that $\mathcal F$ is a $k$-fold composition of $\mathcal H$ if $f = h_1 \circ h_2 \circ \dots \circ h_{k-1} \circ h_k$. Equivalently, we write $f = h^{(k)}$.
\end{definition}

It follows by the conclusion of Lemma \ref{lem:doubler} that $\mathcal G$ is a $2$-fold composition of $\mathcal H$. It is obvious self-composition is associative because the composing functions are identical by definition.

\begin{lemma}
\label{lem:kfoldcomplexity} 
Given a CA $\mathcal F$ that is a  $k$-fold composition of $\mathcal H$, if $X_1^{\mathcal F} = X_1^{\mathcal H}$ then $X_{k n-1}^{\mathcal H} = X_n^{\mathcal F}$. The composite local rule $f$ has radius $kr$.
\end{lemma}

%on $M^{\mathcal H}$
 
\begin{proof}
The $k$-fold composition of $h$ is

\begin{align} \label{eq:kcomposition}
    h^{(k)}(X) : X \;\mapsto\; h\big( \bigxor_{\alpha_1=-r}^{r} h(\bigxor_{\alpha_2=-r}^{r} \hdots  \;h(\bigxor_{\alpha_k=-r}^{r} x_{\alpha_1 + \alpha_2 + \hdots + \alpha_k})\big) \big\} 
\end{align}

\noindent and maps $ X \in \Sigma^{2(k r) + 1} \rightarrow \Sigma$. The proof is given by induction on repeated application of Lemma \ref{lem:doubler}. Instead of two updates ($k=2$), there are $k$ configuration updates per input configuration corresponding to the number of nested applications of $h$ by Eq. \ref{eq:kcomposition}.
\end{proof}

Next, we show that a $k$-fold composition of a CA can be quickly computed:

\begin{lemma}
    Given a CA $\mathcal H$, 
    the $k$-fold composition $\mathcal F$ can be computed in $O(k^2 2^{2kr})$-time. Moreover, $B^{\mathcal F}$ can be constructed with no additional run-time complexity.
    \label{lem:fastFold}
\end{lemma}

\begin{proof}
Note that $r$, the initial rule radius, is constant. It is the time complexity scaling with $k$ that is to be resolved. By Lemma \ref{lem:kfoldcomplexity}, $h^{(k)}$ has radius $kr$. In order to compute $h^{(k)}$, the local rule that governs $\mathcal F$, there are $|\Sigma^{2kr+1}| = 2^{2kr+1}$ configurations (i.e. $(2kr + 1)$-tuples) that need to be iterated over. Each configuration requires $O(k^2)$-time to compute (see the example below).

The corresponding De Bruijn graph $B^{\mathcal F} = (V,E)$ is straightforward to construct. We have that $|V| = 2^{2r}$ vertices. By Definition \ref{def:deB}, the De Bruijn graph $B^{\mathcal F}$ (see Fig. \ref{fig:rule_30_compositions}) has outgoing edges

\[E = \{ v_i \in V ~|~ v_i \rightarrow (v_{2i \text{ mod } |V|},v_{(2i+1) \text{ mod } |V|})\}\]
where index $i$ in binary is equal to the corresponding cell configuration in $\Sigma^{2r}$. The color of each edge is determined by the computation in the first half of this lemma.
\end{proof}

\emph{Example. } We illustrate the time complexity with a concrete example for $r=1$, $k=3$ and $h$ as Rule $30$. This is a $3$-fold composition of Rule 30. Since we have $2(kr) + 1 = 7$, there are $2^7$ possible input configurations. Let $X = \langle \square, \square, \blacksquare, \blacksquare, \square, \blacksquare, \square \rangle \in \Sigma^7$. Applying the rule $h$ once to $X$ in accordance with Definition \ref{def:composition} gives $\langle \blacksquare, \blacksquare, \square, \square, \blacksquare \rangle$. Applying $h$ twice more gives $\langle \square, \blacksquare, \blacksquare \rangle$ and then $\blacksquare$. So $h^{(3)}(X) = \blacksquare$. Clearly, there is a $O(k^2)$ array of cells which have been updated.
\\~

The use of a De Bruijn graph will provide sufficient algorithmic analysis for running a $k$-fold composition.

\begin{lemma}
Let $\mathcal H$ be a CA having a local rule $h$, and radius $r$. Let $X_{n}^{\mathcal H}$ be the current configuration and let $B^{\mathcal H}$ be the De Bruijn graph of $\mathcal H$. Then, $X_{n+1}^{\mathcal H}$ can be computed in $O(n)$-time.
\label{lem:fastRadius}
\end{lemma} 

\begin{proof} To compute $X_{n+1}^{\mathcal H}$, a walk is made on the graph $B^{\mathcal H}$ (see Fig. \ref{fig:walkGraph} for an example where $\mathcal H$ is Rule 30). Each edge traversed is constant-time, where the edge chosen is based on the color of the next cell in the configuration. For any finite configuration length, the walk will begin and end at $\{0\}^{2r}$. Because $X_{n}^{\mathcal H}$ has size $O(n)$, iterating through each color in the configuration to compute $X_{n+1}^{\mathcal H}$ requires $O(n)$-time.
\end{proof}

We are prepared to complete our main result:

\begin{theorem}
\label{thr:faster}
Let $\mathcal H$ be a CA with an arbitrary initial configuration, a local rule $h$, and radius $r$. The configuration $X_{n}^{\mathcal H}$ can be computed in $O( n^2 / \log n)$-time.
\end{theorem}
 
\begin{proof} 
By Lemma \ref{lem:fastFold}, we require $O(k^2 2^{2kr})$-time to generate the $k$-fold composition $h^{(k)}$ of $\mathcal H$. Let $\mathcal F$ by the CA governed by the local rule $h^{(k)}$, and by the same lemma we construct $B^{\mathcal F}$. By Lemma \ref{lem:fastRadius}, we can compute the next generation of any configuration of $\mathcal F$ in $O(n)$. 

From a simple geometric argument it follows that the $k$-fold composition $h^{(k)}$ reduces the time complexity from $O(n^2)$ to $O(n^2 / k)$: for each row that is computed, $k$ are skipped relative to the original rule. Then there exists an optimal maximum $k$ such that the time complexity to run the simulation from the initial configuration to $n$ and compose the rule is equivalent:

\begin{equation} \label{eq:bigequality}
    \frac{n^2}{k} = k^2 2^{2kr}
\end{equation}
up to an arbitrary constant. It is assumed that $n, k$ are much larger than $r$, the initial rule radius. If the time complexities are equal, the total complexity is the complexity of either operation multiplied by a constant factor. The solution to the equation is then the Lambert $W$ function (also called the ``product logarithm")

\begin{equation} \label{eq:k_exact}
    k(n) = \frac{3}{2 r \ln (2)} W_0 \bigg( \frac{2r \ln(2)}{ 3 } n^{2/3}\bigg) \sim \log n
\end{equation}
by using $W_0(n) \sim \log n$ for large $n$ and simplifying. Thus it is shown that the time complexity is dominated by the exponential number of states for a given $k$. In fact, the prefactor $k^2$ in the cost to generate the composite rule output can be any positive polynomial and not change the asymptotic time complexity.
\end{proof}

\begin{corollary}
The configuration $X_n^{\mathcal H}$ satisfying the conclusion of Theorem \ref{thr:faster} can be computed using $O(n^2 / (\log n)^3)$-space.
\label{cor:memUsage}
\end{corollary}

\begin{proof} Rearranging Eq. \ref{eq:bigequality} to solve for the number of states $|\Sigma^{2 k r+1} | = 2^{2 k r + 1} \sim n^2 / k^3$ and in the asymptotic limit of large $n$ we arrive at $\sim n^2 / (\log n)^3$. 
\end{proof}

\section{Experimental Results} 
In this section, we present experimental results of the method described in Theorem \ref{thr:faster}, managing to remove a $\log$ factor from the base run-time complexity of trivially running an ECA. The results are presented in Fig. \ref{fig:experiment_times} and the code is given in the Appendix. Its practicality on the general-purpose computer is limited by two things: 1. Memory access is not sequential and so is not cache-friendly 2. Random-access memory (RAM) is limited. As it turns out, the bitwise optimizations described by Wolfram \cite{Wolfram2002} are several times faster than this method for any reasonable number of generations because computers parallelize packed bitwise operations and are efficient at reading memory linearly along an array.

\begin{figure}[!h]
    \centering
    \includegraphics[width=\textwidth]{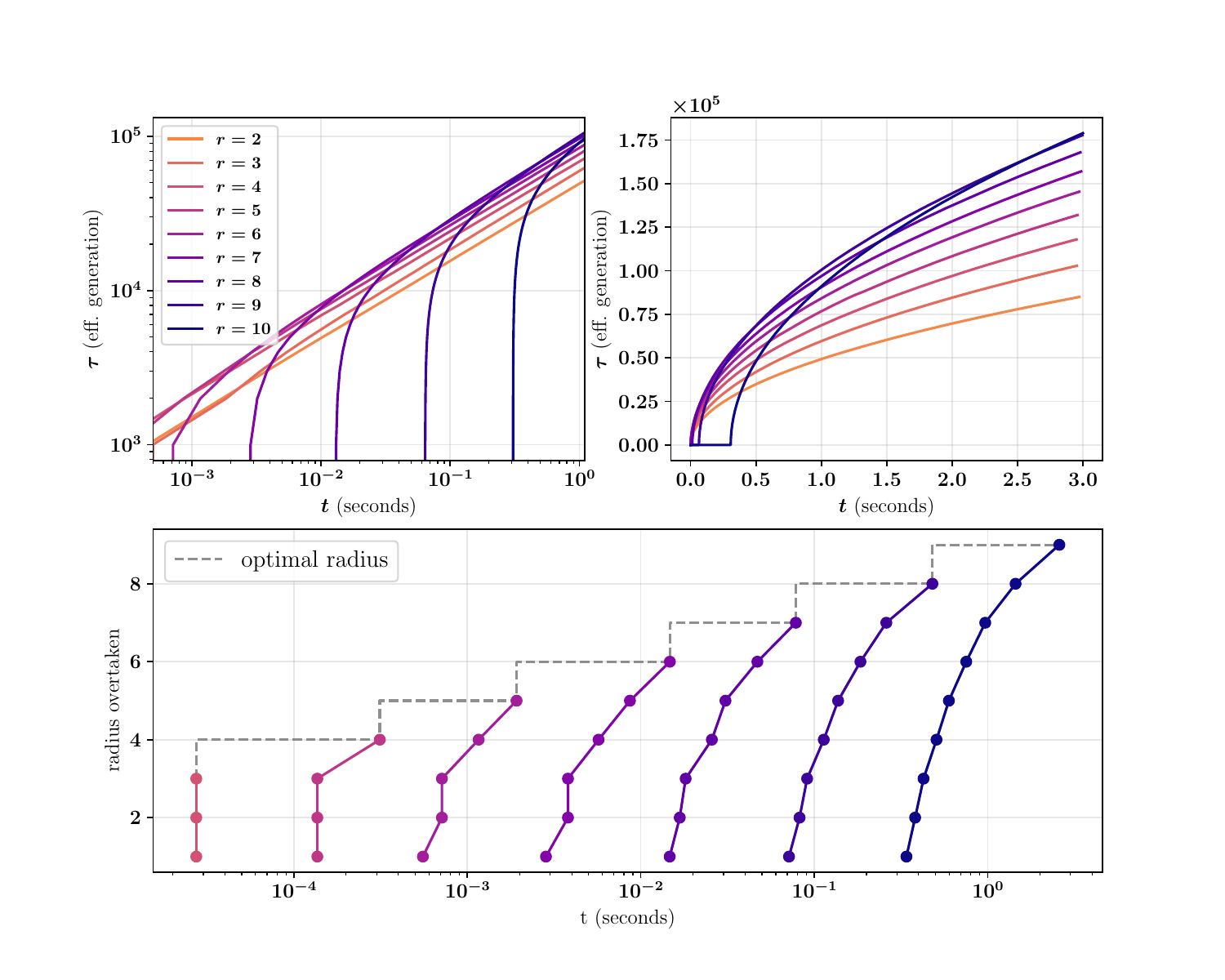}
    \caption{Constructing a larger radius automaton improves simulation speed on a given machine (Intel(R) Xeon(R) Gold 6230 CPU @ 2.10GHz). The bottom panel shows the optimal radius as a function of time. It also shows the times at which a simulation at a given radius surpasses those of a smaller radius given on the vertical axis. Given a fixed amount of simulation time $t$ seconds, the optimal radius is approximated by $\lfloor 0.37 \log_2 (t) + 9.4 \rfloor$.
    }
    \label{fig:experiment_times}
\end{figure}

\begin{figure}
    \centering
    \includegraphics[width=0.85
    \linewidth]{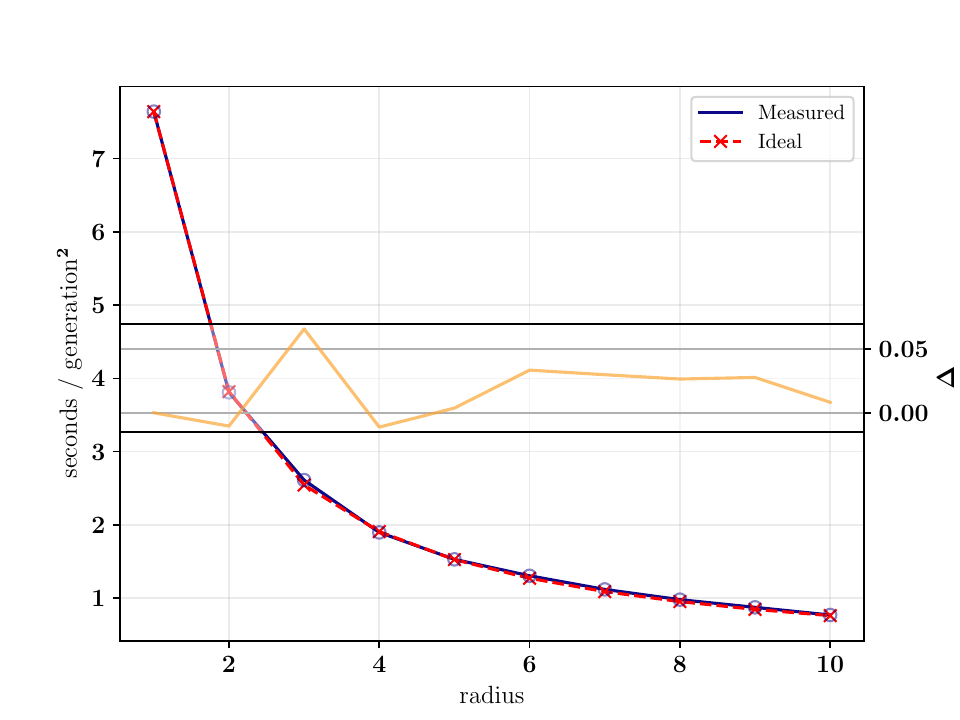}
    \caption{The machine in this experiment obeys an $n^2 / k$ time complexity scaling law ($k=r$) to compute the next generation, validating the use of Eq. \ref{eq:bigequality}. The ideal curve uses $r=1$ as a reference, so if $f(r)$ is the measured number of seconds per squared generation, the ideal is  $f_{\text{ideal}}(r)=f(r=1)/r$.}
    \label{fig:idealram}
\end{figure}

\begin{figure}[]
    \centering
    \includegraphics[width=0.8\linewidth]{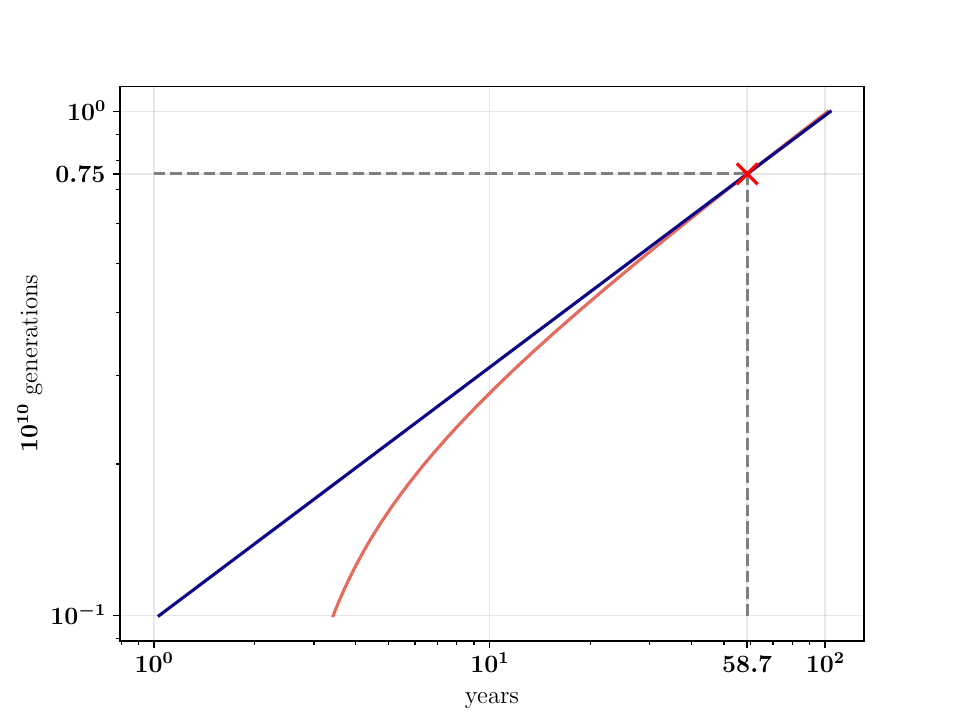}
    \caption{The $27$-fold composition (4.5 petabytes) will overtake the bitwise-optimized implementation in about $\sim 60$ years on an Intel(R) Xeon(R) Gold CPU. This takes the principle of delayed gratification to its extreme. This algorithm might only be practical with special hardware. The curves were generated by extrapolating the quadratic time versus generation curves and the exponential dependence on $r$ for composite rule creation.}
    \label{fig:overtaken}
\end{figure}

These memory problems could be reduced if contiguous memory is indexed in a De Bruijn sequence \cite{Weisstein} of order $2r+1$ on the alphabet $\Sigma$. This is because, given a rule input $\langle \blacksquare, \square, \blacksquare \rangle$, the adjacent one in the configuration is either $\langle \square, \blacksquare, \blacksquare \rangle$ or $\langle \square, \blacksquare, \square \rangle$ which is a bit-shift left and rewriting of the right-most bit (flipped depending on the endianness). If the index in memory is $q$, then the next transition will map to memory location $\lfloor q / 2 \rfloor$ or  $\lfloor q / 2 \rfloor + 2^{2r}$ which may be a large jump. For $2$-color cellular automata with an equal probability of being black or white, there would be an equal chance that the next state would be in the neighboring memory address. 

Nonetheless, Fig. \ref{fig:idealram} shows that computers access memory at speeds independent of address size in the plotted regime. And despite the aforementioned memory inefficiencies, creating large compositions would eventually overtake the bitwise-optimized implementation, as shown in Fig. \ref{fig:overtaken}. 

\section{Discussion}

This result is a single point on the domain of space-time functional dependencies that are possible for computing an arbitrary generation of an elementary cellular automaton, and its method is straightforward. On a computer, it relies on a model of computation that treats accessing memory with an address of size $O(\log n)$ as taking a single unit of time---the canonical RAM model. This model is a convention brought about by advances in physical machinery \cite{KNUTH1973189}. One might argue that a proper machine must be a boolean circuit that uses the standard basis of boolean operations. But in this computational model of circuit complexity, the interpretation of our result is clear: our algorithm increases circuit size while decreasing circuit depth in accordance with our equations. The operation of traversing an edge on the De Bruijn graph is realizable in constant circuit depth, which is what we would define as time in that model. With the contemporary interest in interaction nets, graph rewriting, and category theory, perhaps a different model of computation will become standard. 

There may exist faster machines to compute certain cellular automata, and maybe some that can be described by a simple equation over the integers. Open questions remain about the possibility of such formal machines, their potential for practical instantiation, and their implications for emergent phenomena in complex systems. 

\newpage
\section{Appendix}

In this Appendix, we show several De Bruijn graph examples and provide code implementing our algorithm.

\definecolor{sutnerred}{RGB}{196,77,77}
\definecolor{sutnergreen}{RGB}{77,196,77}

\begin{figure}[h]
    \centering
    \includegraphics[width=0.9\linewidth]{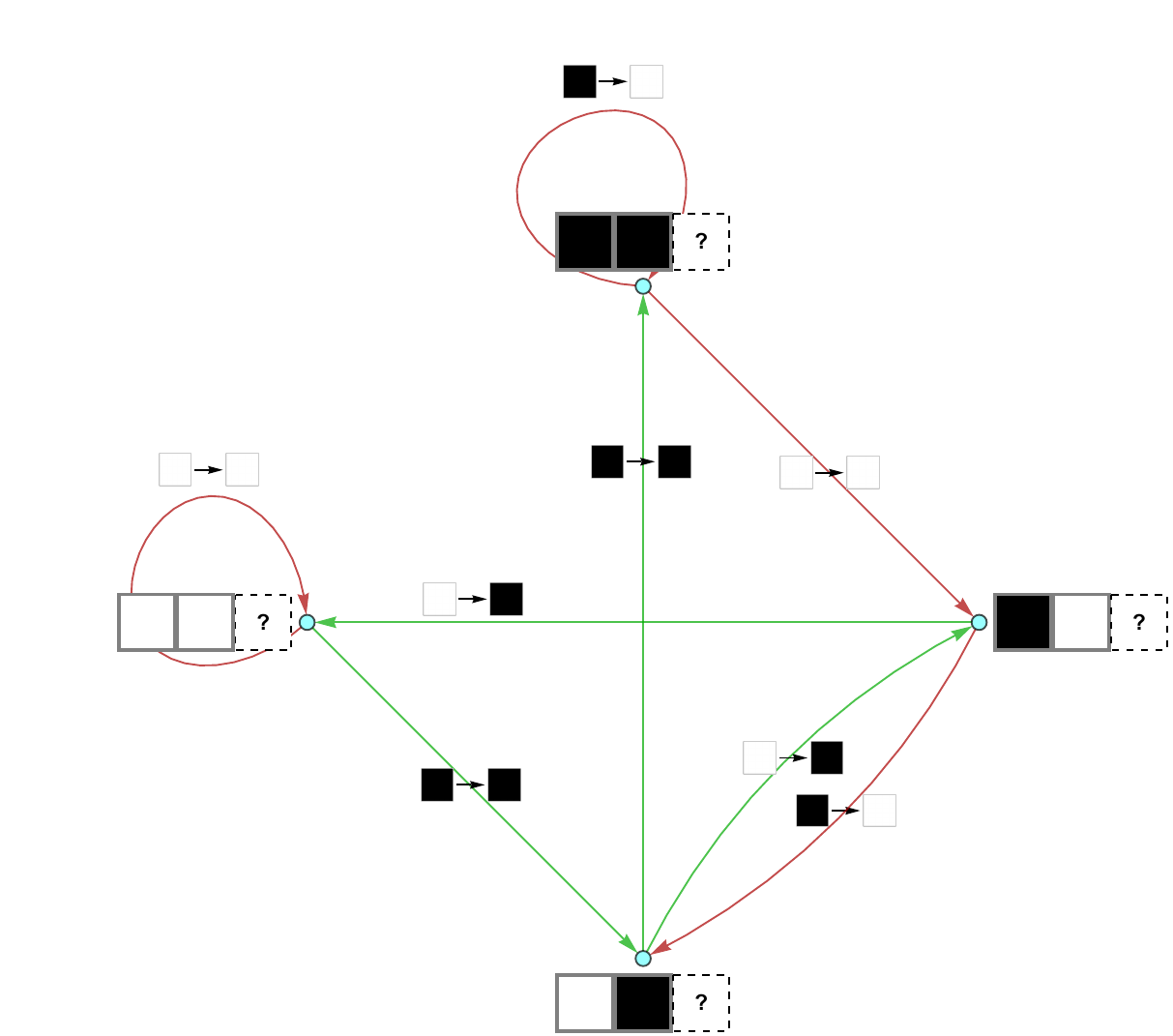}
    \caption{Rule 30's state transition (colored De Bruijn) diagram. The left cell in the edge rule $\{ \square,\blacksquare \} \rightarrow \{ \square,\blacksquare \}$ is read in \textit{from} the cell configuration, and the right cell is written \textit{to} the configuration. Red edges visually indicate this output cell is $\blacksquare$ and green edges indicate the output is $\square$. As an edge is traversed, the neighborhood is shifted left by one cell. 
    }
    \label{fig:state_transition_graph}
\end{figure}

\begin{figure}[h]
    \centering

    \includegraphics[width=\linewidth]{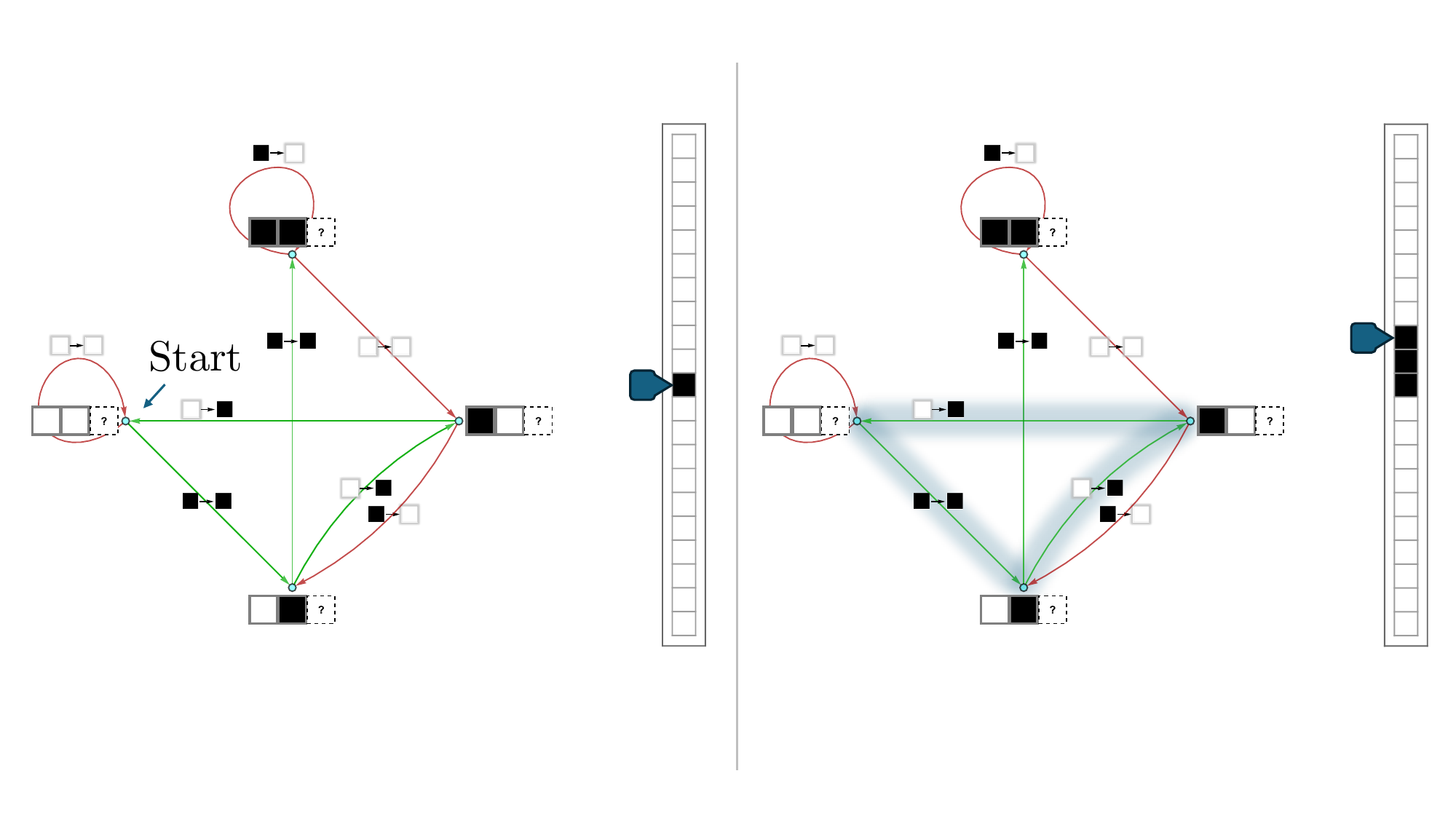}
    \caption{Computing the next generation of the ECA Rule 30 is taking a walk in the state transition graph. The figure on the right shows the walk taken in the graph in order to produce $\langle \blacksquare,\blacksquare,\blacksquare \rangle$.}
    \label{fig:walkGraph}
\end{figure}

\begin{figure}[!h]
    \centering
    \includegraphics[width=\textwidth]{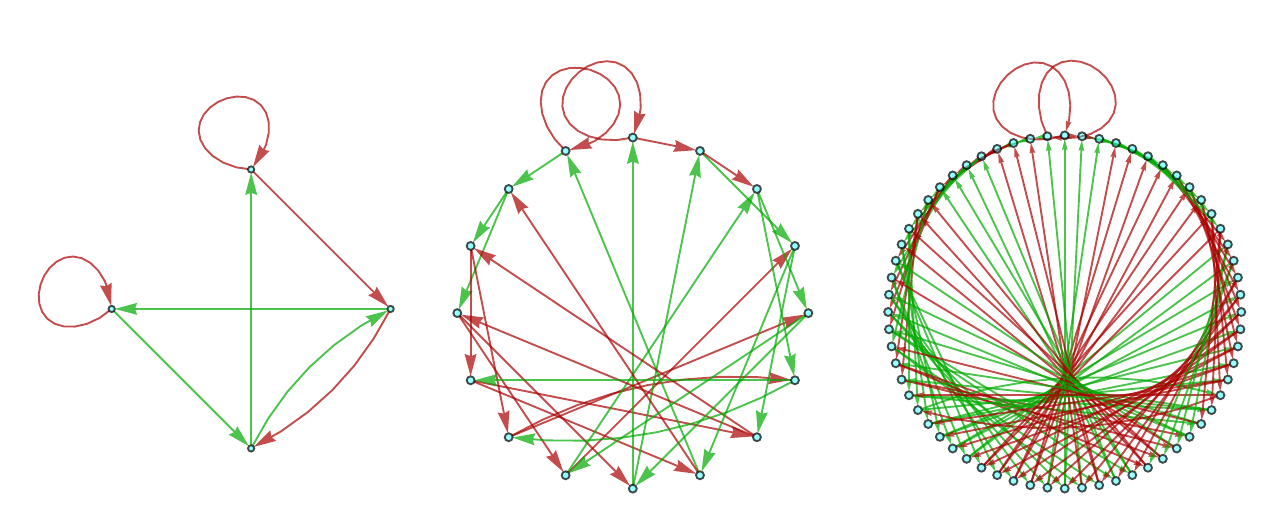}
   \caption{
   1-fold, 2-fold, and 3-fold composition state transition diagrams of the Rule $30$ ECA in circular embeddings \cite{Sutner}. They have a regular structure, with any vertex at index $i$ along the circular embedding having outgoing connections to vertices at $2i$ and $2i+1$ modulo the number of vertices.}
    \label{fig:rule_30_compositions}
\end{figure}

\newpage
\subsection{Experiment code}
The following code can be compiled using the C++17 standard. Alternatively, it can be implemented in a few lines of compiled Wolfram Language code.

\lstdefinestyle{mystyle}{
backgroundcolor=\color{white},   % choose the background color
basicstyle=\footnotesize,        % the size of the fonts that are used for the code
breakatwhitespace=false,         % sets if automatic breaks should only happen at whitespace
breaklines=true,                 % sets automatic line breaking
captionpos=b,                    % sets the caption-position to bottom
commentstyle=\color{ForestGreen},    % comment style
deletekeywords={...},            % if you want to delete keywords from the given language
escapeinside={\%*}{*)},          % if you want to add LaTeX within your code
extendedchars=true,              % lets you use non-ASCII characters; for 8-bits encodings only, does not work with UTF-8
keepspaces=true,                 % keeps spaces in text, useful for keeping indentation of code (possibly needs columns=flexible)
keywordstyle=\color{Blue},       % keyword style
language=C++,                 % the language of the code
morekeywords={*,...},            % if you want to add more keywords to the set
numbers=none,                    % where to put the line-numbers; possible values are (none, left, right)
showspaces=false,                % show spaces everywhere adding particular underscores; it overrides 'showstringspaces'
showstringspaces=false,          % underline spaces within strings only
showtabs=false,                  % show tabs within strings adding particular underscores
stepnumber=2,                    % the step between two line-numbers. If it's 1, each line will be numbered
stringstyle=\color{red},     % string literal style
tabsize=2,                       % sets default tabsize to 2 spaces
}

\begin{lstlisting}[language=C++,style=mystyle]
#include <iostream>
#include <vector>
#include <string>
#include <chrono>
#include <array>

#define N 10000000
unsigned int composite_rule_r = 3;

// create bit masks of length 2r + 1 for truncating configurations
constexpr std::array<uint64_t, 30>  create_masks() {
    std::array<uint64_t, 30> table = {0};
    for (unsigned int i = 1; i < 30; ++i) {
        table[i] = (1UL << (2*i + 1)) - 1;
    }
    return table;
}

constexpr std::array<uint64_t, 30> masks = create_masks();

std::vector<unsigned char> config;

unsigned int config_vec_size;
unsigned int config_width;
unsigned int i;
unsigned int tau;

class Timer {
private:
    using Clock = std::chrono::steady_clock;
    using Second = std::chrono::duration<double, std::ratio<1>>;
    std::chrono::time_point<Clock> m_beg { Clock::now() };

public:
    void reset() {
        m_beg = Clock::now();
    }

    double elapsed() const {
        return std::chrono::duration_cast<Second>(Clock::now() - m_beg).count();
    }
};

void pad_config(){
    config.insert(config.begin(), 6000, 0);
    config.insert(config.end(), 6000, 0);
    config_vec_size = config.size();
}

inline void run_simulation() {

    const unsigned int r = 1;

    Timer t;
    unsigned char config_ptr;

    for (; tau < N; ++tau) {
        
        unsigned int j;
        
        if (2*(config_width + 2*r + 1) >= config_vec_size)
            pad_config();

        config_ptr=0;
        for (j = config_vec_size / 2 - config_width; j < config_vec_size / 2 + config_width; j++){
            config_ptr <<= 1;
            config_ptr |= config[j];
            //explicit boolean equation for rule 30 in terms of XOR and OR
            config[j-1] = ((config_ptr>>2)&1)^(((config_ptr>>1)&1) | (config_ptr&1));
        }
        
        config_width += r;   
    
        if (tau % 500 == 0)
            std::cout << t.elapsed() << '\t' << tau << std::endl;
    }
}

inline void run_precomputed_simulation() {

    //rule 30 lookup
    std::vector<unsigned char> init_rule = {0,1,1,1,1,0,0,0};
    
    std::vector<unsigned char> composite_rule;

    Timer t;
    
    std::cout << t.elapsed() << '\t' << tau << std::endl;
    
    unsigned int j;
    unsigned int k;
    unsigned int l;

    //precompute the composite rule starting from r = 1
    const unsigned int r = composite_rule_r;
    uint64_t local_config[2] = {0,0};
    composite_rule.resize(1UL << (2*r + 1));

    for (j = 0; j < (1UL << (2*r+ 1)); ++j){ 
        
        local_config[0] = j;

        for (k  = 0; k < r; ++k){
            
            local_config[(k+1)%2] = 0;

            for (l = 0; l < 2*r + 1 - 2*k; ++l){
                local_config[(k+1)%2] |= static_cast<uint64_t>(init_rule[(local_config[k%2] >> l) & 0b111]) << l;
            }
        }

        composite_rule[j] = local_config[k%2] & 1UL;
    }

    //give sufficient padding when starting from a simple seed
    config_width = 3*r;

    std::cout << t.elapsed() << '\t' << tau << std::endl;
    uint64_t config_ptr = 0;

    //run the simulation
    for (; tau < N; ++i, tau += r){

        if (2*(config_width + (2*r + 1)) >= config_vec_size)
            pad_config();
        
        config_ptr = 0;
        for (j = config_vec_size/2 - config_width; j < config_vec_size/2 + config_width; j++) {
            config_ptr <<= 1;
            config_ptr |= config[j];
            config_ptr &= masks[r];
            config[j - r] = composite_rule[config_ptr];            
        }

        config_width += r;   

        if (i % (500/r) == 0)
            std::cout << t.elapsed() << '\t' << tau << std::endl;
    }
}

int main(int argc, char* argv[]) {

    if (argc < 2){
        std::cout << "Enter a composite rule radius or 0 for the reference: " << std::endl;
        std::string input;
        std::cin >> input;
        composite_rule_r = std::stoi(input);
    }
    else
        composite_rule_r = std::stoi(argv[1]);
        
    i = 1;
    tau = 1;
        
    config.resize(6000);
    config_vec_size = config.size();

    //initialize the config with a simple seed
    config[config_vec_size/2] = 1;

    if (composite_rule_r == 0)
        run_simulation();
    else
        run_precomputed_simulation();    

    return 0;

}
\end{lstlisting}

\printbibliography

\end{document}